\documentclass{article}
\usepackage[a4paper]{geometry}

\usepackage{authblk}
\usepackage{amsmath} 
\usepackage{xspace} 
\usepackage{amsthm} 
\usepackage{todonotes}
\usepackage{xargs} 
\usepackage{url}
\usepackage[linesnumbered,ruled,vlined]{algorithm2e} 
\usepackage[autolanguage,np]{numprint}
\usepackage[hidelinks]{hyperref}
\usetikzlibrary{hobby,backgrounds,calc,trees,fit}

\newtheorem{definition}{Definition}[section]
\newtheorem{theorem}{Theorem}[section]
\newtheorem{lemma}{Lemma}[section]

\newcommand{\myred}{red!80!black}
\newcommand{\mygreen}{green!80!black}
\newcommand{\myblue}{blue!80!black}

\newcommand{\etal}{{\it et al}\xspace}
\newcommand{\ie}{{\it i.e.}\xspace}
\newcommand{\etc}{{\it etc.}\xspace}

\newcommand{\cpp}{{\tt C++}\xspace}

\newcommand{\B}{\mathcal{B}}
\newcommand{\Nc}{N^{C}}
\newcommand{\Ec}{E^{C}}
\newcommand{\Gc}{G^{C}}
\newcommand{\Np}{N^{P}}

\newcommand{\lra}{\Leftrightarrow}
\renewcommand{\l}{\ell}

\newcommand{\bu}{{b_U}}
\newcommand{\bv}{{b_V}}

\newcommand{\du}{{d_U}}
\newcommand{\dv}{{d_V}}
\newcommand{\ddu}{{{d_2}_U}}
\newcommand{\ddv}{{{d_2}_V}}
\newcommand{\zu}{{\zeta_U}}
\renewcommand{\nu}{{n_U}}
\newcommand{\nv}{{n_V}}
\DeclareMathOperator*{\argmin}{argmin}

\newcommand{\return}{{\bf return}\xspace}
\SetKwProg{Fn}{Function}{:}{}
\SetKwFunction{BKMBE}{BipBronKerbosch}
\SetKwFunction{BK}{BronKerbosch}
\newcommand{\each}{{\bf each}\xspace}

\newcommand{\bbk}{\textsc{bbk}\xspace}
\newcommand{\oombea}{\textsc{oombea}\xspace}

\newcommand{\ucforum}{{\it UC-Forum}\xspace}
\newcommand{\discogs}{{\it Discogs}\xspace}
\newcommand{\citeseer}{{\it CiteSeer}\xspace}
\newcommand{\marvel}{{\it Marvel}\xspace}
\newcommand{\dbpedia}{{\it DBpedia}\xspace}
\newcommand{\actormovie}{{\it Actor-Movie}\xspace}
\newcommand{\pics}{{\it Pics}\xspace}
\newcommand{\wikilens}{{\it WikiLens}\xspace}
\newcommand{\bookcrossing}{{\it BookCrossing}\xspace}
\newcommand{\github}{{\it GitHub}\xspace}
\newcommand{\dailykos}{{\it DailyKos}\xspace}
\newcommand{\filmtrust}{{\it FilmTrust}\xspace}
\newcommand{\citeulike}{{\it CiteULike}\xspace}
\newcommand{\reuters}{{\it Reuters}\xspace}
\newcommand{\bibsonomy}{{\it BibSonomy}\xspace}
\newcommand{\tvtropes}{{\it TV-Tropes}\xspace}
\newcommand{\dvdciao}{{\it DVD-Ciao}\xspace}
\newcommand{\youtube}{{\it YouTube}\xspace}
\newcommand{\nips}{{\it NIPS-Papers}\xspace}
\newcommand{\movie}{{\it MovieLens}\xspace}

\newcommand{\convexpath}[2]{
[   
create hullnodes/.code={
\global\edef\namelist{#1}
\foreach [count=\counter] \nodename in \namelist {
\global\edef\numberofnodes{\counter}
\node at (\nodename) [draw=none,name=hullnode\counter] {};
}
\node at (hullnode\numberofnodes) [name=hullnode0,draw=none] {};
\pgfmathtruncatemacro\lastnumber{\numberofnodes+1}
\node at (hullnode1) [name=hullnode\lastnumber,draw=none] {};
},
create hullnodes
]
($(hullnode1)!#2!-90:(hullnode0)$)
\foreach [
evaluate=\currentnode as \previousnode using \currentnode-1,
evaluate=\currentnode as \nextnode using \currentnode+1
] \currentnode in {1,...,\numberofnodes} {
let
\p1 = ($(hullnode\currentnode)!#2!-90:(hullnode\previousnode)$),
\p2 = ($(hullnode\currentnode)!#2!90:(hullnode\nextnode)$),
\p3 = ($(\p1) - (hullnode\currentnode)$),
\n1 = {atan2(\y3,\x3)},
\p4 = ($(\p2) - (hullnode\currentnode)$),
\n2 = {atan2(\y4,\x4)},
\n{delta} = {-Mod(\n1-\n2,360)}
in 
{-- (\p1) arc[start angle=\n1, delta angle=\n{delta}, radius=#2] -- (\p2)}
}
-- cycle
}

\begin{document}


\title{BBK: a simpler, faster algorithm for enumerating maximal bicliques in large sparse bipartite graphs}

\author[1]{Alexis Baudin}
\author[1]{Clémence Magnien}
\author[1]{Lionel Tabourier}
\affil[1]{Sorbonne Université, CNRS, LIP6, F-75005 Paris \protect\\
  {\tt firstname.lastname@lip6.fr}}

\date{}

\maketitle


\begin{abstract}
  Bipartite graphs are a prevalent modeling tool for real-world networks, capturing interactions between vertices of two different types.
Within this framework, bicliques emerge as crucial structures when studying dense subgraphs:
they are sets of vertices such that all vertices of the first type interact with all vertices of the second type.
Therefore, they allow identifying groups of closely related vertices of the network, such as individuals with similar interests or webpages with similar contents.
This article introduces a new algorithm designed for the exhaustive enumeration of maximal bicliques within a bipartite graph.
This algorithm, called \bbk for \emph{Bipartite Bron-Kerbosch}, is a new extension to the bipartite case of the Bron-Kerbosch algorithm, which enumerates the maximal cliques in standard (non-bipartite) graphs.
It is faster than the state-of-the-art algorithms and allows the enumeration on massive bipartite graphs that are not manageable with existing implementations. 
We analyze it theoretically to establish two complexity formulas: one as a function of the input and one as a function of the output characteristics of the algorithm.
We also provide an open-access implementation of \bbk in \cpp, which we use to experiment and validate
its efficiency on massive real-world datasets and show that its execution time is shorter in practice than state-of-the art algorithms.
These experiments also show that the order in which the vertices are processed, as well as the choice of one of the two types of vertices on which to initiate the enumeration have an impact on the computation time.

\end{abstract}


\section{Introduction}

Bipartite graphs are widely used to represent real-world networks~\cite{guillaume2006bipartite}. 
They can model many systems where two different types of entities interact.
They are thus widely used to describe social systems such as online platforms where users select content (watch videos, click on links, buy products, \etc)~\cite{zhou2007bipartite,zhu2015measuring}, or individuals taking part in projects or events~\cite{newman2001scientific,li2014name,coates2020transfer}.
It is also a popular representation for biological systems~\cite{huber2007graphs,pavlopoulos2018bipartite}, or for ecological networks~\cite{dormann2009indices,simmons2019motifs,chabert2022impact}.
As with non-bipartite graphs, the identification of dense subgraphs in these networks is important for analyzing their structure and understanding their functioning~\cite{broder2000graph,yoon2007co,lehmann2008biclique}: 
for example, it can reveal users with common interests~\cite{muhammad2016summarizing}, or information about the organization of proteins~\cite{bu2003topological,voggenreiter2012exact}.
Also, as bipartite graphs can represent sets of items, with one type of nodes representing baskets or sets and the other the items themselves~\cite{makino2004new,zaki1998theoretical},  enumerating dense subgraphs in bipartite graphs has a close connection with mining frequent itemsets in databases, a long-standing task in data mining~\cite{borgelt2012frequent}, with various applications such as finding association rules in large databases~\cite{agrawal1994fast}.

A bipartite graph is a triplet $G = (U,V,E)$, where $U$ and $V$ are two sets of disjoint vertices, and $E$ is a set of edges between elements of $U$ and elements of $V$: $E \subseteq U \times V$.
Graphs are undirected, so there is no distinction between an edge $(u,v)$ and $(v,u)$.
Throughout this paper, if $A$ is a set of vertices of $G=(U,V,E)$, then we denote by $A_U$ the set of vertices of $A$ that are in $U$, \ie $A \cap U$, and $A_V$ the set of vertices of $A$ that are in $V$: $A \cap V$. 
A biclique of $G$ is a set $C \subseteq U \cup V$ such that the vertices of $C_U$ are all connected to the vertices of $C_V$. 
It is said to be maximal when it is not included in any other biclique.
An example is given in Figure~\ref{fig:ex-graphebip}.

\begin{figure}[!hbtp]
  \centering
  
  \begin{tikzpicture}[
    every node/.style={circle, draw, fill=black!30, minimum width=8mm, minimum height=6mm},
    ]
    
    \node (3) at (0, 0) {3};
    \node (2) at (-3, 0) {2};
    \node (1) at (-6,0) {1};
    \node (E) at (1, -2) {E};
    \node (D) at (-1,-2) {D};
    \node (C) at (-3, -2) {C};
    \node (B) at (-5, -2) {B};
    \node (A) at (-7, -2) {A};
    \draw (1) -- (A);
    \draw (1) -- (B);
    \draw (B) -- (2);
    \draw (B) -- (3);
    \draw (2) -- (C);
    \draw (2) -- (D);
    \draw (C) -- (3);
    \draw (D) -- (3);
    \draw (E) -- (3);

    \begin{pgfonlayer}{background}

      \foreach \nodename in {1,2,3,A,B,C,D,E} {
        \coordinate (\nodename') at (\nodename);
      }

      \path[fill=\myblue,opacity=0.2] \convexpath{B,3,E,D,C}{14pt};
      \path[draw=\myblue, dashed, line width=1pt, opacity=0.8] \convexpath{B,3,E,D,C}{14pt};
      
      \path[fill=\mygreen,opacity=0.2] \convexpath{A,1,B}{20pt};
      \path[draw=\mygreen, dashed, line width=1pt, opacity=0.8] \convexpath{A,1,B}{20pt};

      \path[fill=\myred,opacity=0.2] \convexpath{B,2,3,D,C}{19pt};
      \path[draw=\myred, dashed, line width=1pt, opacity=0.8] \convexpath{B,2,3,D,C}{19pt};
      
    \end{pgfonlayer}
  \end{tikzpicture}

  \caption{Example of a bipartite graph, with three maximal bicliques circled in color: $\{1,A,B\}$, $\{2,3,B,C,D\}$ and $\{3,B,C,D,E\}$.
    Note that this graph has two other maximal bicliques, $\{1,2,3,B\}$ and $\{A,B,C,D,E\}$, not represented here for the sake of clarity.
  }
  \label{fig:ex-graphebip}
\end{figure}
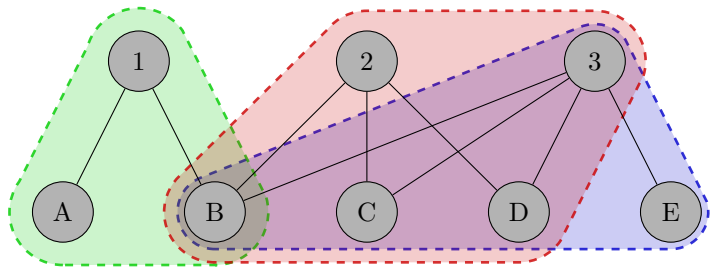

The enumeration of bicliques, particularly in large bipartite graphs, has been the subject of much work. 
Recent works have notably improved the state-of-the-art of maximal biclique enumeration by adapting the famous Bron-Kerbosch algorithm~\cite{bron1973algorithm} to this context.
In this paper, we propose the \bbk algorithm, which is also an adaptation of the Bron-Kerbosch algorithm to bipartite graph, but in a way that we believe to be simpler and proves to be more time-efficient than the current standards. We provide an open-access C++ implementation of this algorithm\,\footnote{\url{https://gitlab.lip6.fr/baudin/bbk}}.

The rest of the article is organized as follows: in Section~\ref{sec:related} we present the state of the art of maximal biclique enumeration in bipartite graphs; in Section~\ref{sec:algorithm} we introduce and detail our new algorithm \bbk; in Section~\ref{sec:complexity} we formalize complexities of this algorithm, on the one hand as a function of its input characteristics and on the other hand as a function of its output characteristics; 
in Section~\ref{sec:experiments} we carry out an extensive experimental study of the \bbk algorithm, comparing its implementation with the best in the state of the art to validate its efficiency on massive real datasets, and showing to what extent the choice of one of the two vertex sets as a starting point for \bbk has an impact on computation time; 
finally, in Section~\ref{sec:conclusion} we conclude the study and present the perspectives of this work.

\section{Related work}
\label{sec:related}

In the literature, numerous works have been devoted to enumerating maximal bicliques.
Algorithms for exhaustive enumeration of maximal bicliques in bipartite graphs were introduced and experimented in the 2000s~\cite{makino2004new}. They have since been improved several times. In 2014, Damaschke~\etal.~\cite{damaschke2014enumerating} developed an algorithm that improves the enumeration on bipartite graphs with a heterogeneous degree distribution, justified by the fact that this is the case for real-world graphs, with a theoretical complexity suited to the case of these particular graphs. 
At the same period, Zhang~\etal.~\cite{zhang2014finding} implemented the first adaptation of the Bron-Kerbosch algorithm~\cite{bron1973algorithm}, which is the reference algorithm for enumerating maximal cliques in non-bipartite graphs.
Using this new algorithm, they handled larger datasets than what was previously possible, in particular biological datasets.
Later, in 2018, Das~\etal.~\cite{das2018shared} proposed a parallel algorithm to enumerate maximal bicliques on massive datasets.
In 2020, Qin~\etal.~\cite{qin2020efficient} developed an enumeration method which is supposed to be well suited to unbalanced graphs, however they did not provide an implementation of their algorithm. 
Then, Abidi~\etal.~\cite{abidi2020pivot} improved the adaptation of the Bron-Kerbosch algorithm of Zhang~\etal.~\cite{zhang2014finding} by reducing the search space using a pivot,
which is a technique to reduce the number of recursive calls of the main function and therefore speed up the execution.
Hriez~\etal. proposed in 2021 an interesting method to achieve the enumeration~\cite{hriez2021effective}, which consists in adding edges (a process known as graph inflation) to simplify the algorithm. However, as described the method cannot properly scale to massive graphs from the real world.
Finally, Chen~\etal.~\cite{chen2022efficient} proposed another improvement on the algorithm of Abidi~\etal.~\cite{abidi2020pivot}, which performs much better than all the methods of the state-of-the-art.
For this reason, we take Chen~\etal. method, called \oombea, as a reference for comparison to our own maximal biclique enumeration algorithm throughout this article.

As it plays an important role in our study, we give additional details about \oombea.
It is inspired by the Bron-Kerbosch algorithm~\cite{bron1973algorithm}.
It is a recursive algorithm that identifies a given biclique $C$ by its subset of $C_U$: the vertices in $C_V$ are those that are neighbors to all vertices in $C_U$.
The function then considers the {\em candidate} vertices $u \in U$ such that there are maximal bicliques containing $C_U \cup \{u\}$ that have not already been enumerated, and makes a recursive call on each of these vertices.
Besides that, the authors also seem to take inspiration from the work of Eppstein~\etal.~\cite{eppstein2010listing}, which is one of the most efficient implementation of Bron-Kerbosch algorithm for maximal clique enumeration on real-world non-bipartite graphs.
Eppstein~\etal. introduced the idea of processing vertices according to a degeneracy order~\cite{batagelj2003m} to reduce 
the number of candidate vertices on which recursive calls are made.
Similarly, Chen~\etal. considered one of the two sets of vertices of the bipartite graph, which we denote by $U$ in this paragraph, and enumerate the maximal bicliques starting from each vertex of $U$ according to a given order which, like Eppstein~\etal., tends to reduce the size of the set of candidate vertices.
The order that \oombea uses is a degeneracy order of the bipartite graph projected onto $U$,
\ie the graph whose vertex set is $U$ and two vertices are linked if they have a common neighbor in $V$.
Chen~\etal. showed that the complexity of this algorithm is in~$\mathcal{O}\left(\nu \cdot \zu \cdot m \cdot 2^{\zu}\right)$, where $\nu$ is the number of vertices in $U$, $\zu$ is the degeneracy of the graph projected onto $U$ and $m$ is the number of edges in $G$. 
They also expressed the complexity of \oombea in terms of the number of maximal bicliques $\B$ as $\mathcal{O}\left(\zu \cdot m \cdot \B\right)$. 
Note that even for a sparse graph, the graph projected onto $U$ can be dense, and therefore can have a high $\zu$ value.

\section{New maximal biclique enumeration algorithm}
\label{sec:algorithm}

In this section, we use a graph-inflation process to transform the bipartite graph and then
apply directly the Bron-Kerbosch algorithm~\cite{bron1973algorithm} to enumerate maximal bicliques.
We can then use the version of this algorithm from Eppstein~\etal.~\cite{eppstein2010listing},
which processes the vertices in a specific order to improve the computation time for massive real-world graphs. 
To do this, we define what we call a \emph{bidegeneracy order} of vertices, adapted from the degeneracy order proposed by Eppstein~\etal.
This provides a simple way to perform the maximal biclique enumeration.
However, this first approach is inefficient, so we use properties of neighborhood in bipartite graphs to enhance its performances in the \bbk algorithm.

\subsection{Clique-extended graph of a bipartite graph}

To extend the enumeration of maximal cliques in graphs to the enumeration of maximal bicliques in bipartite graphs, we define the \emph{clique-extended graph} of a bipartite graph
by adding edges between all the vertices of $U$, and between all the vertices of $V$.
We call this graph $\Gc$ and define it formally below.

\begin{definition}[Clique-extended graph of a bipartite graph]
  The \emph{clique-extended graph of the bipartite graph} $G$ is the graph $\Gc = (U \cup V, \Ec)$, where: 
  $$\Ec{} = E \cup \{(u_1,u_2) \in U^2 \ | \ u_1 \neq u_2\} \cup \{(v_1,v_2) \in V^2 \ | \ v_1 \neq v_2\}.$$
\end{definition}

This clique-extended graph induces a notion of neighborhood, which we call \emph{clique-extended neighborhood}, noted~$\Nc$:

\begin{definition}[Clique-extended neighborhood of a vertex]
  Let $x \in U \cup V$. The \emph{clique-extended neighborhood of $x$} corresponds to the neighbors of $x$ in the clique-extended graph $\Gc$. It is denoted by $\Nc(x)$:
  $$\Nc(x) =
  \begin{cases}
    N(x) \cup U \setminus \{x\} \text{ if } x \in U\\
    N(x) \cup V \setminus \{x\} \text{ if } x \in V.
  \end{cases}
  $$
\end{definition}

The clique-extended graph has a particular property, that we exploit for the \bbk algorithm: a set of vertices that forms a (maximal) clique of $\Gc$ equivalently forms a (maximal) biclique of $G$. 
This result was introduced by Gély~\etal.~\cite{gely2009enumeration}:

\begin{theorem}
  Let $G = (U,V,E)$ be a bipartite graph. Then the maximal cliques of $\Gc$ correspond to the maximal bicliques of $G$:
  $$C \text{ is a maximal clique of } \Gc \lra C \text{ is a maximal biclique of } G.$$
  \label{thm:graphe-clique-eq}
\end{theorem}

This theorem induces a direct method for enumerating the maximal bicliques of a bipartite graph: it is sufficient to enumerate the maximal cliques of its extended graph. 
Algorithm~\ref{algo:bkmbe-etendu} proposes the pseudocode for this method: it follows the Eppstein~\etal. algorithm~\cite{eppstein2010listing}.
\begin{algorithm}[!h]
  \DontPrintSemicolon
  \KwIn{Bipartite graph $G = (U,V,E)$.} 
  \KwOut{Set of maximal bicliques of $G$.}
  \For {\each $x_i$ in a degeneracy order $x_1, x_2, \dots, x_n$ of $U \cup V$\label{line:bkmbe-e:for-order}}
  {
    $P_i \gets \Nc(x_i) \setminus \{x_1, \dots, x_{i-1}\}$ \label{line:bkmbe-e:Pinit}\;
    $X_i \gets \Nc(x_i) \cap \{x_1, \dots,x_{i-1}\}$ \;
    \BK{$\{x_i\}$, $P_i$, $X_i$} \label{line:bkmbe-e:initBK}\;
  }
  \Fn{\BK{$R$, $P$, $X$}}{
    \If{$P \cup X = \emptyset$\label{line:bkmbe-e:maxtest}} { 
      \return $R$ maximal biclique \label{line:bkmbe-e:return} \;
    }
    $p \gets$ pivot in $P \cup X$ \label{line:bkmbe:pivot} \;
    \For {$x \in P \setminus \Nc(p)$\label{line:bkmbe:elagage}} { 
      \BK{$R \cup \{x\}$, $P \cap \Nc(x)$, $X \cap \Nc(x)$} \label{line:bkmbe-e:RC} \;
      $P \gets P \setminus \{x\}$ \;
      $X \gets X \cup \{x\}$ \label{line:bkmbe:updateX}\;
    }
  }
  \caption{Enumerate the maximal bicliques using the extended graph.}
  \label{algo:bkmbe-etendu}
\end{algorithm}

\subsection{BBK: a new algorithm for maximal biclique enumeration}

Algorithm~\ref{algo:bkmbe-etendu} is straightforward, but it cannot be used in practice for graphs containing many vertices, as the sets of candidates $P_i$ defined at Line~\ref{line:bkmbe-e:Pinit} 
can be larger than $U$ or $V$.
To overcome this issue and provide an algorithm usable efficiently on sparse, massive, bipartite graphs, we develop a revised version which takes advantage of the bipartite nature of the graph.
This allows to reduce the size of the candidate sets and to use $N$ instead of $\Nc$ in the main function, which induces much fewer neighbors.
In addition, we refine the biclique maximality test to perform it earlier in the process; 
we also limit the enumeration to bicliques containing a vertex $u$ for each $u \in U$ instead of browsing all $U \cup V$;
finally, we order the vertices to improve the efficiency of the enumeration by introducing the notion of a \emph{bidegeneracy order}.
Each of these points is detailed in the rest of this section.

We call this new algorithm \bbk, and its pseudocode is given in Algorithm~\ref{algo:bbk}. 
Let us first summarize the workings of \bbk before going into its details:
\begin{itemize}
\item Lines~\ref{line:bbk:forU} to~\ref{line:bbk:initBK} give a more efficient way of initializing biclique enumerations on each vertex, to be compared with Lines~\ref{line:bkmbe-e:for-order} to~\ref{line:bkmbe-e:initBK} of Algorithm~\ref{algo:bkmbe-etendu};
\item Lines~\ref{line:bbk:ifPUPVempty} to~\ref{line:bbk:stopenum} are an improvement on the maximality test for bicliques performed in Lines~\ref{line:bkmbe-e:maxtest} to~\ref{line:bkmbe-e:return} of Algorithm~\ref{algo:bkmbe-etendu};
\item finally, Lines~\ref{line:bbk:pivot} to~\ref{line:bbk:updateX} are equivalent to Lines~\ref{line:bkmbe:pivot} to~\ref{line:bkmbe:updateX} of Algorithm~\ref{algo:bkmbe-etendu}, but adapted to use the neighborhood $N$ in the bipartite graph, instead of $ \Nc $.
\end{itemize}

\begin{algorithm}[!hbtp]
  \DontPrintSemicolon
  \KwIn{Bipartite graph $G = (U,V,E)$.} 
  \KwOut{Set of maximal bicliques of $G$.}
  \tcp{More efficient initialization:}
  \For {\each $u_i$ in a bidegeneracy order $u_1, u_2, \dots, u_n$ \mbox{of } $U$\label{line:bbk:forU}}
  {
    $P_i \gets \Np(u_i) \setminus \{u_1, \dots, u_{i-1}\}$ \label{line:bbk:Pui} \;
    $X_i \gets \Np(u_i) \cap \{u_1, \dots,u_{i-1}\}$ \label{line:bbk:Xui} \;
    \BKMBE{$\{u_i\}$, $P_i$, $X_i$} \label{line:bbk:initBK} \;
  }
  
  \Fn{\BKMBE{$R$, $P$, $X$}}{
    \tcp{Maximality test:}
    \If{($P_U = \emptyset$ {\bf or} $P_V = \emptyset$) {\bf and} $X = \emptyset$\label{line:bbk:ifPUPVempty}} { 
      \return $R \cup P$ maximal biclique \label{line:bbk:returnbiclique} \;
    }
    \If {($P_U = \emptyset$ {\bf and} $X_V \neq \emptyset$) {\bf or} ($P_V = \emptyset$ {\bf and} $X_U \neq \emptyset$)\label{line:bbk:maxtest2}} {
      \return \label{line:bbk:stopenum} \;
    }

    \tcp{\BK adapted operations to leverage bipartite nature of the graph:}
    $p \gets$ pivot in $P \cup X$ \label{line:bbk:pivot} \;
    \If {$p \in U$\label{line:bbk:ifpU}} {
      $Q \gets P_V \setminus N(p)$ \label{line:bbk:pruneU}\;
    }
    \Else {
      $Q \gets P_U \setminus N(p)$ \label{line:bbk:pruneV}\;
    }
    \If{$p \in P$} {
      $Q \gets \{p\} \cup Q$ \label{line:bbk:Qp} \;
    }
    \For {$x \in Q$\label{line:bbk:forxQ}} {
      \If {$x \in U$} {
        $P_x \gets (P_U \setminus \{x\}) \cup (P_V \cap N(x))$ \;
        $X_x \gets X_U \cup (X_V \cap N(x))$ \label{line:bbk:XxU}\;
      }
      \Else {
        $P_x \gets (P_V \setminus \{x\}) \cup (P_U \cap N(x))$ \;
        $X_x \gets X_V \cup (X_U \cap N(x))$ \;
      }
      \BKMBE{$R \cup \{x\}$, $P_x$, $X_x$} \label{line:bbk:RC}\;
      $P \gets P \setminus \{x\}$ \;
      $X \gets X \cup \{x\}$ \label{line:bbk:updateX} \;
    }
  }
  \caption{\bbk: Bron-Kerbosch adapted to maximal biclique enumeration.}
  \label{algo:bbk}
\end{algorithm}

\subsubsection{Lines~\ref{line:bbk:forU} to~\ref{line:bbk:initBK}: efficient initialization of the biclique enumeration}
\label{subsubsec:bidegen}

\paragraph{Set of vertices for the initialization}

The initialization of calls to \BK performed at Line~\ref{line:bkmbe-e:initBK} of Algorithm~\ref{algo:bkmbe-etendu} can be  improved by performing the following two operations.

Firstly, it is not necessary to enumerate the set of maximal bicliques that contain $x$ for each vertex $x \in U \cup V$. 
Instead, we can simply list the maximal bicliques that contain a vertex $u \in U$. 
Indeed, $V$ is the only biclique that contains no vertices in $U$ and that can be maximal; it can therefore be added to the enumeration outside the core of the algorithm. 
Thus, in Algorithm~\ref{algo:bbk}, the loop starting at Line~\ref{line:bbk:forU} is only performed on $U$.
Note that this idea has already been used by Chen~\etal.~\cite{chen2022efficient} in \oombea.

Secondly, to enumerate the maximal cliques in a graph, the Bron-Kerbosch algorithm uses the fact that the neighbors of $u$ are the vertices belonging to a clique which contains $u$.
In a bipartite graph, this is not the case: if $u \in U$, the set of vertices that are in a biclique $C$ which contains $u$, when $C_V \neq \emptyset$, are \emph{the vertices of $N(u) \cup N_2(u)$} where $N_2(u)$ is the set of neighbors of $u$'s neighbors excluding $u$ itself.
Note that this observation was made by Hermelin and Manoussakis~\cite{hermelin2021efficient}.
Therefore, we formalize below this particular neighborhood of $u$ as it plays an important role in the description of the \bbk algorithm.

\begin{definition}[Projection-extended neighborhood]
  \label{def:Np}
  Let $u \in U$. We call the vertices of $N(u) \cup N_2(u)$
  the \emph{projection-extended neighborhood of $u$} and denote it by $\Np(u)$, where $N_2(u)$ is the set of neighbors of $u$'s neighbors excluding $u$ itself.
\end{definition}

We use this projection-extended neighborhood by searching, for each vertex $u \in U$, the maximal bicliques that contain $u$ among the vertices of the set $\Np(u)$.
Thus, the sets $\Nc(u_i)$ at Line~\ref{line:bkmbe-e:for-order} of Algorithm~\ref{algo:bkmbe-etendu} are replaced by the sets $\Np(u_i)$ at Line~\ref{line:bbk:forU} of Algorithm~\ref{algo:bbk}, which are much smaller in practice (see Section~\ref{sec:experiments}).

Note that this reasoning ignores the biclique $U$ which is the only one which does not contain any vertex in $V$.
As above, this biclique can be added to the enumeration outside the core of the algorithm if it is maximal.

\paragraph{Bidegeneracy order of vertices}
Eppstein~\etal.~\cite{eppstein2010listing} have shown that the order of vertices has a significant impact on the enumeration efficiency on non-bipartite real-world graphs.
They use a degeneracy order to reduce the maximum size of the candidate vertex sets $P_i$ 
on which recursive calls are made.
We extend this concept by introducing the notion of a \emph{bidegeneracy order} on $U$:

\begin{definition}[Bidegeneracy order of $U$]
  A \emph{bidegeneracy order of $U$} is an order $u_1,\dots,u_n$ such that $u_i$ is a vertex of $U$ of minimal number of projection-extended neighbors in the subgraph induced by the vertices $u_i, u_{i+1}, \ldots, u_n$ and their neighbors. In other words, for all~$i \in \{1,\dots,n\}$, 
  $$u_i = \argmin_{u \in \{u_i, \ldots, u_n\}} \left( |\Np(u) \setminus \{u_1,\dots,u_{i-1}\}| \right).$$
\end{definition}

Such an order is obtained by iteratively selecting an unselected vertex $u \in U$ that minimizes $\Np(u)$, then updating the sets in $\Np$ by deleting $u$.
The objective of using a bidegeneracy order of $U$ is to reduce the maximum size of the candidate sets $P_i = \Np(u_i) \setminus \{u_1,\dots,u_{i-1}\}$ on which the enumeration is performed at Line~\ref{line:bbk:initBK} of Algorithm~\ref{algo:bbk}.
To quantify this maximum size, we introduce below the notions of the \emph{bidegeneracy of a vertex} and the \emph{bidegeneracy of a set of vertices}.

\begin{definition}[Bidegeneracy of a vertex]
  The \emph{bidegeneracy of a vertex} $u \in U$ is the maximum value $b$ such that there exists $U' \subseteq U$ with $u \in U'$ verifying $\forall x \in U', \ \left| \Np(x) \cap (U' \cup V) \right| \geq b$.
  We denote it by $b(u)$.
  If $u \in V$, its bidegeneracy is defined symmetrically by inverting $U$ and $V$.
\end{definition}

\begin{definition}[Bidegeneracy of $U$ and $V$]
  The bidegeneracy of $U$, denoted by $\bu$, is the maximum bidegeneracy of the vertices of $U$. The bidegeneracy of $V$, denoted $\bv$, is defined similarly on $V$.
\end{definition}

For example, let us consider the graph of Figure~\ref{fig:ex-graphebip} with $U = \{A,B,C,D,E\}$ and $V = \{1,2,3\}$. 
In this example,
$\Np(A) = \{1, B\} $, so if we set $U' = \{A,B\}$, then $\forall x \in U', \ |\Np(x) \cap (U' \cup V)| \geq 2$; moreover there is no set which yields a larger value, thus $ b(A) = 2$.
We can show similarly that the bidegeneracy of  $B$, $C$, $D$ and $E$ is $4$, thus $b_U = 4$.

Thanks to the use of a bidegeneracy order, we can show that the size of a candidate set $P_i = \Np(u_i) \setminus \{u_1,\dots,u_{i-1}\}$ at Line~\ref{line:bbk:Pui} of Algorithm~\ref{algo:bbk} 
is at most $b(u_i)$ and therefore the maximum size of this set over all vertices in $U$ is $\bu$. 
Indeed, when $u_i$ is selected following such an order, it is a vertex of minimal number of projection-extended neighbors in the subgraph induced by $u_i,u_{i+1},\dots,u_n$ and their neighbors.
In other words, if we set $U'=\{u_i, \ldots, u_n\}$, then
$\forall u \in U', |\Np(u) \setminus \{u_1,\dots,u_{i-1}\}|
= |\Np(u) \cap (U' \cup V)| \ge |\Np(u_i) \setminus \{u_1,\dots,u_{i-1}\}|$.
Since the bidegeneracy of $u_i$ is the largest value over all sets $U'$ 
satisfying the inequality above,
we obtain that $b(u_i) \ge |\Np(u_i) \setminus \{u_1,\dots,u_{i-1}\}|$.
This improvement is efficient in that the bidegeneracy is much smaller in practice than the maximum value of $|\Np(u)|$ for $u \in U$ (see Section~\ref{sec:experiments}).
This is particularly important as the enumeration within one of these sets is exponential in the size of that set, as detailed in Section~\ref{sec:complexity}.

\subsubsection{Lines~\ref{line:bbk:ifPUPVempty} to~\ref{line:bbk:stopenum}: improving the biclique maximality test}

The maximality test performed at Line~\ref{line:bkmbe-e:maxtest} of Algorithm~\ref{algo:bkmbe-etendu} can be performed earlier by taking into account the bipartite nature of the graph. 
Indeed, if $P_U = \emptyset$, then there are two cases which require no recursive call on any vertex of $P_V$ and that can be tested in constant time:

\begin{itemize}
\item If $X = \emptyset$, then $R_U \cup R_V \cup P_V$, which is a biclique by construction, is maximal and can therefore be output without making further recursive calls
  (Lines~\ref{line:bbk:ifPUPVempty} - \ref{line:bbk:returnbiclique}).
\item If $X_V \neq \emptyset$, then $X_V$ cannot be modified again, as $X_V$ is only modified at Lines~\ref{line:bbk:XxU} when a vertex of $P_U$ is added to the clique under construction. 
  Therefore, $P_U = \emptyset$ and $X_V \neq \emptyset$ for all the subcalls launched at Line~\ref{line:bbk:RC}, and then none of them can lead to a maximal biclique (which would require $X_V$ to be empty).
  Thus, the call can be stopped there (Lines~\ref{line:bbk:maxtest2} - \ref{line:bbk:stopenum}).
\end{itemize}

The same observations can be made when swapping the roles of $U$ and $V$.
Altogether, Lines~\ref{line:bkmbe-e:maxtest} to~\ref{line:bkmbe-e:return} in Algorithm~\ref{algo:bkmbe-etendu} are replaced by Lines~\ref{line:bbk:ifPUPVempty} to~\ref{line:bbk:stopenum} in Algorithm~\ref{algo:bbk}.

\subsubsection{Lines~\ref{line:bbk:pivot} to~\ref{line:bbk:updateX}:  using $N$ instead of $\Nc$}

The neighborhoods $\Nc$ in Algorithm~\ref{algo:bkmbe-etendu}  are usually too large to be processed efficiently on massive graphs.
Fortunately, it is not necessary to store and manipulate them in practice.
Indeed, it is possible to compute the sets handled in Lines~\ref{line:bkmbe:pivot} to~\ref{line:bkmbe:updateX} of Algorithm~\ref{algo:bkmbe-etendu} by considering only the bipartite neighborhood $N$.
In this sense, we show that Lines~\ref{line:bbk:pivot} to~\ref{line:bbk:updateX} of Algorithm~\ref{algo:bbk} are equivalent
to those lines.

When $x \in U$, the set $P \cap \Nc(x)$ (Line~\ref{line:bkmbe-e:RC} of Algorithm~\ref{algo:bkmbe-etendu}) is equal to $P \cap (N(x) \cup U \setminus \{x\})$, thus, it is equal to $(P_U \setminus \{x\}) \cup (P_V \cap N(x))$.
Symmetrically, when $x \in V$, $P \cap \Nc(x)$ is equal to $(P_V \setminus \{x\}) \cup (P_U \cap N(x))$.
The same applies to $X \cap \Nc(x)$, following the same reasoning.
This leads to the sets $P_x$ and $X_x$ defined within the loop starting at Line~\ref{line:bbk:forxQ} of Algorithm~\ref{algo:bbk}.

We can apply the same reasoning to the pruning of the pivot occurring at Line~\ref{line:bkmbe:elagage} of Algorithm~\ref{algo:bkmbe-etendu}. 
Indeed, if $p \in U$, $P \setminus \Nc(p) = P \setminus (N(p) \cup U \setminus \{p\})
= (P_U \setminus (U \setminus \{p\})) \cup (P_V \setminus N(p))
= \{p\} \cup (P_V \setminus N(p))$ if $p \in P$, and $P \setminus \Nc(p) = P_V \setminus N(p)$ if $p \notin P$ (Lines~\ref{line:bbk:pruneU} and~\ref{line:bbk:Qp}).
Similarly, if $p \in V$ then $P \setminus \Nc(p) = \{p\} \cup (P_U \setminus N(p))$ 
if $p \in P$, and $P \setminus \Nc(p) = P_U \setminus N(p)$ if $p \notin P$ (Lines~\ref{line:bbk:pruneV} and~\ref{line:bbk:Qp}).
Note in particular that the pivot prunes all vertices (except itself) on its own side, in addition to removing vertices from $N(p)$ in the other side, which is more efficient than only pruning the vertices of $N(p)$.

As with the maximal clique enumeration in non-bipartite graphs, the pivot is chosen at Line~\ref{line:bbk:pivot} to maximize the number of vertices pruned,
\ie to minimize the size of $Q$.

\section{Complexity of BBK algorithm}
\label{sec:complexity}

In this section, we express the complexity of Algorithm \bbk as a function of its input and output characteristics (resp. Section~\ref{subsec:complexity_input} and Section~\ref{subsec:complexity_output}).
To do so, we use that Algorithm \bbk has been inspired by Eppstein \etal. work on maximal clique enumeration~\cite{eppstein2010listing}.
Thus, we can derive its complexity as a function of its input following an approach similar to theirs,
and the complexity as a function of its output following later works~\cite{conte2022overall,baudin2023faster}.

In what follows, we denote by $d_U$ the maximum degree of a vertex of $U$, and by $d_V$ the maximum degree of a vertex of~$V$.

\subsection{Complexity as a function of input characteristics}

\label{subsec:complexity_input}

Before expressing the complexity as a function of the input in Theorem~\ref{thm:bbk1}, let us first introduce the following preliminary Lemma~\ref{lemma:bbk}.

\begin{lemma}
  The bidegeneracy of $U$ is larger than the maximum degrees of $U$ and $V$:
  $$\du \leq \bu \text{ and } \dv \leq \bu.$$
  \label{lemma:bbk}
\end{lemma}

\begin{proof}
  Let $u \in U$ be a vertex of degree $d(u) = \du$, and consider $U' = \{u\}$, then $\forall x \in U', \ |\Np(x) \cap (U' \cup V)| = \du$. 
  Thus, $b(u) \geq \du$, hence $\bu \geq \du$.

  Now let $v \in V$ be a vertex of degree $d(v) = \dv$, and consider $U' = N(v)$, then since each vertex of $N(v)$ is adjacent to $v$, we deduce that $\forall x \in U', \ |\Np(x) \cap (U' \cup V)| \geq \dv$.
  Thus, for any $u \in U'$, $b(u) \geq \dv$ hence $\bu \geq \dv$.
\end{proof}

\begin{theorem}
  \label{thm:bbk1}
  The complexity of Algorithm~\ref{algo:bbk} \bbk is in~$\mathcal{O}\left((\nu + m) \cdot \bu \cdot 3^{\bu/3}\right)$, where $\nu$ is the number of vertices in $U$, $m$ is the number of edges and $\bu$ is the bidegeneracy of $U$.
\end{theorem}

\begin{proof}  
  The algorithm first computes a bidegeneracy order. 
  This requires, for each vertex $u \in U$, to calculate its projection-extended neighborhood $\Np(u)$ in~$\mathcal{O}\left(\du \cdot \dv\right)$. 
  Then, the bidegeneracy order is computed by iteratively taking a vertex $u$ of minimal $|\Np(u)|$ and decreasing by $1$ the value of $|\Np(u')|$ for each $u' \in \Np(u) \cap U$, and there are at most $\du \cdot \dv$ such nodes. 
  The whole procedure is thus carried out in~$\mathcal{O}\left(\nu \cdot \du \cdot \dv\right)$, that is in $\mathcal{O}\left(\nu \cdot \bu^2\right)$ according to Lemma~\ref{lemma:bbk}.

  Then, for each vertex $u_i \in U$ in the resulting order, the algorithm enumerates the maximal bicliques containing $u_i$ and no vertex preceding it in the order, using the \BKMBE function at Line~\ref{line:bbk:initBK}. 
  It begins with computing the sets $P_i = \Np(u_i) \setminus \{u_1,\dots,u_{i-1}\}$ and $X_i = \Np(u_i) \cap \{u_1,\dots,u_{i-1}\}$ at Lines~\ref{line:bbk:Pui} and~\ref{line:bbk:Xui}, in $\mathcal{O}\left(\du \cdot \dv\right)$. 
  Then, to evaluate the cost of the call to \BKMBE at Line~\ref{line:bbk:initBK}, we can here apply the complexity expression of the Eppstein~\etal. algorithm~\cite{eppstein2010listing} in the case of a non-bipartite graph, as it can be noticed that the cost of the operations performed are the same.
  To do so, we can use their Lemma~3.6 where the authors show that given $c \geq |P_i|$, this call is done in~$\mathcal{O}\left((c + |X_i|) \cdot 3^{c/3}\right)$, when choosing a pivot that maximizes the number of cut vertices.
  So here the complexity of the call is in $\mathcal{O}\left((b(u_i) + |X_i|) \cdot 3^{b(u_i)/3}\right)$, as by  definition of a bidegeneracy order, we know that the size of the candidate set $P_i$ is at most $b(u_i)$ (see Section~\ref{subsubsec:bidegen}). 
  Therefore, the loop at Line~\ref{line:bbk:forU} is done in~$\mathcal{O}\left(\underset{i=1}{\overset{\nu}{\sum}} (b(u_i) + |X_i|) \cdot 3^{|P_i|/3}\right)$, that is in~$\mathcal{O}\left(\left( \nu \cdot \bu + \underset{i=1}{\overset{\nu}{\sum}} |X_i| \right) \cdot 3^{\bu/3}\right)$.
  Now,  $\underset{i=1}{\overset{\nu}{\sum}} |X_i| \leq \underset{u \in U}{\sum} |\Np(u)| \leq \underset{u \in U}{\sum} |N(u)| + \underset{u \in U}{\sum} |N_2(u)|$ (see Definition~\ref{def:Np} that defines $\Np$).
  Besides, $\underset{u \in U}{\sum} |N(u)| = m$ and $\underset{u \in U}{\sum} |N_2(u)| \leq \underset{u \in U}{\sum} \underset{v \in N(u)}{\sum} d(v) \leq \underset{u \in U}{\sum} (\dv \cdot |N(u)|) = \dv \cdot m$. Thus, $\underset{i=1}{\overset{\nu}{\sum}} |X_i|$ is in~$\mathcal{O}\left(\dv \cdot m\right)$, that is in~$\mathcal{O}\left(\bu \cdot m\right)$ according to Lemma~\ref{lemma:bbk}. Then, the loop at Line~\ref{line:bbk:forU} is done in~$\mathcal{O}\left((\nu + m) \cdot \bu \cdot 3^{\bu/3}\right)$.

  Finally, the \bbk algorithm runs in~$\mathcal{O}\left((\nu + m) \cdot \bu \cdot 3^{\bu/3} + \nu \cdot \bu^2\right)$. 
  In addition, for any integer $k \geq 0$, $k^2 \leq k \cdot 3^{k/3}$, so in particular $\bu^2 \leq \bu \cdot 3^{\bu/3}$ which leads to the complexity expression in the statement.
\end{proof}

This complexity should be compared with that of Chen~\etal.'s \oombea algorithm~\cite{chen2022efficient}. 
They show that it runs in~$\mathcal{O}\left(m \cdot \nu \cdot \zu \cdot  2^{\zu}\right)$, where $\zu$, called \emph{unilateral coreness}, is the degeneracy of the graph projected onto $U$, and $m$ is the number of edges in $G$.
The bidegeneracy and the unilateral coreness are closely related concepts, with $\bu \geq \zu$.
The factor $(\nu+m) \cdot \bu \cdot 3^{\bu/3}$ in our complexity therefore corresponds to
the factor $m \cdot \zu \cdot 2^\zu$ in the \oombea algorithm complexity, whereas this second complexity features an additional factor $\nu$.

Several practical observations can be made about the complexity expression of
Theorem~\ref{thm:bbk1}.
First, the bidegeneracy of $U$ plays a central role in the complexity, due to the exponential factor in $\bu$, which also points out the benefit of using a bidegeneracy order.
Indeed, this order makes it possible to bound by $\bu$ the maximum size of a candidate set $P_i = \Np(u_i) \setminus \{u_1,\dots,u_{i-1}\}$, while a random order would lead to bound it by $ \du+\ddu $, where $\ddu$ is the maximum degree of the graph projected onto $U$.
Furthermore, we show in Section~\ref{sec:experiments} (see Table~\ref{tab:dmax}) that the maximal bidegeneracy is in practice close to its optimal value, \ie the maximal degree of the graph (Lemma~\ref{lemma:bbk}). 
In other words, without the bidegeneracy order, there would be an additional $\ddu$ in the bound, and this would be of less benefit as the projected degree is much larger in practice than the bidegeneracy.

It should be also noted that this complexity in $\mathcal{O}\left(\nu \cdot \bu \cdot 3^{\bu/3}\right)$ is more precisely in $\mathcal{O}\left(\underset{i=1}{\overset{\nu}{\sum}} (b(u_i) + |X_i|) \cdot 3^{b(u_i)/3}\right)$, as seen in the proof of Theorem~\ref{thm:bbk1}.
It is a relevant point, as $b(u_i)$ of most vertices $u_i \in U$ is well below the maximum value $\bu$ in massive real-world networks,
as can be seen in Table~\ref{tab:dmax}: the average bidegeneracy in real-world datasets is usually much lower than the maximum bidegeneracy.
However, the exponential factor $3^{\bu/3}$ in the complexity formula can be high on real data, even in cases where \bbk  exhibits good performance in practice (see Section~\ref{sec:experiments}).
That is because this worst case complexity can be far from the actual running time of the algorithm, which is why we develop in the following subsection a complexity expression as a function of the algorithm output characteristics.

\subsection{Complexity as a function of output characteristics}
\label{subsec:complexity_output}

Now, we formulate the complexity of Algorithm~\ref{algo:bbk} \bbk as a function of its output characteristics, precisely:

\begin{itemize}
\item $\boldsymbol{\B}$: the number of maximal bicliques in $G$;
\item $\boldsymbol{q}$: the maximal size of a non-trivial biclique (meaning that $C_U \neq \emptyset$ and $C_V \neq \emptyset$);
\item $\boldsymbol{d}$: the maximal degree in $G$;
\item $\boldsymbol{\du}$: the maximal degree of vertices in $U$, \ie $\du = max_{u \in U} \{ |N(u)| \}$;
\item $\boldsymbol{\ddu}$: the maximal degree of the graph projected onto $U$, \ie  $\ddu = max_{u \in U} \{ |N_2(u)| \}$
\end{itemize}

To do this, we consider the trees of recursive calls made by function~\BKMBE within Algorithm~\ref{algo:bbk}. 
The initializing call of this function is made at Line~\ref{line:bbk:initBK}, and recursive calls are made at Line~\ref{line:bbk:RC}.
The internal nodes of these trees correspond to calls for which the set of vertices
$Q$ on which iterates the loop on Line~\ref{line:bbk:forxQ} is not empty, \ie they generate other child calls, while the leaves correspond to calls that generate no other.

Inspired by the work of Conte \etal.~\cite{conte2022overall}, we focus in what follows on the leaves of these call trees, which we separate into two categories: those that output a maximal clique and those that do not.
The latter correspond to unnecessary computations, as they do not contribute to the enumeration. An optimal pivot pruning strategy would cut the branches that lead to these leaves, leaving only leaves that return a maximal biclique. Let us note $\l$ the total number of leaves in the call trees. 
Some of these leaves return maximal bicliques (counted in $\B$), and others do not.
We are then interested in the ratio of ``good'' leaves:

$$\boldsymbol{r = \dfrac{\B}{\l}}.$$

In particular, if $r$ is less than 1, it means that there are unnecessary recursive calls.
Using this ratio, we establish Theorem~\ref{thm:bbk2} to express the complexity of Algorithm~\ref{algo:bbk} as a function of its output.

\begin{theorem}
  \label{thm:bbk2}
  Using the above definition, we have $1 \leq \frac{1}{r} \leq 2^q$, and the complexity of Algorithm~\ref{algo:bbk} is in~$\mathcal{O}\left(\frac{1}{r} \cdot (\du + \ddu) \cdot d \cdot q \cdot \B\right)$.
\end{theorem}

\begin{proof}
  First, we show that $1 \leq \frac{1}{r} \leq 2^q$. 
  On the one hand, it is clear that $\frac{1}{r} \geq 1$ by definition of $r$. 
  On the other hand, each maximal biclique of $G$ contains at most $2^q$ sub-bicliques, so there are at most $2^q \cdot \B$ bicliques in total in the graph. 
  Now, observe that the set $R$ associated to a node of the call trees is a biclique,
  and the root call to the \BKMBE function at iteration $i$ of the loop at Line~\ref{line:bbk:forU} enumerates all bicliques $R$ that contain $x_i$ and none of the vertices $x_1,\dots,x_{i-1}$.
  Consequently, the root call of each iteration enumerates a set of bicliques $R$
  that is disjoint from the sets of bicliques resulting from other iterations.
  The same applies to the recursive calls made in the loop at Line~\ref{line:bbk:forxQ}: as each vertex processed in an iteration is placed in $X$, no biclique $R$ in subsequent iterations can contain that vertex.
  So, each node in the call trees of \BKMBE is associated to a  biclique different from any other node in any other call tree, so that there are at most $2^q \cdot \B$ nodes in the trees.
  Thus, as each leaf is a particular node of a tree, we deduce that $\l \leq 2^q \cdot \B$, and therefore $\frac{1}{r} \leq 2^q$.

  Now, we express the complexity of the \bbk algorithm. 
  By definition of $q$, we know that the depth of any call tree is at most $q$, thus there are at most $q \cdot \l$ nodes in the call trees. 
  Besides, the number of vertices in the set $P \cup X$ that can augment the current biclique is in~$\mathcal{O}\left(\du + \ddu\right)$.
  The pivot is therefore chosen in $\mathcal{O}\left((\du + \ddu) \cdot d\right)$, by calculating the size of $P \cap N(p)$ for each $p \in P \cup X$. 
  Furthermore, the intersections of the sets $P_U$, $P_V$, $X_U$ and $X_V$ with $N(x)$ within the loop starting at Line~\ref{line:bbk:forxQ} are done in $\mathcal{O}\left(d\right)$, and this loop iterates over at most $|P|$ vertices, so it runs in $\mathcal{O}\left((\du + \ddu) \cdot d\right)$ too.

  So, the complexity of Algorithm~\ref{algo:bbk} is in $\mathcal{O}\left((\du + \ddu) \cdot d \cdot q \cdot \l\right)$. 
  As $\l = \frac{1}{r} \cdot \B$, this complexity can be expressed as~$\mathcal{O}\left(\frac{1}{r} \cdot (\du + \ddu) \cdot d \cdot q \cdot  \B\right)$. 
\end{proof}

This expression of the complexity as a function of output characteristics gives an insight on how close \bbk is from an optimal enumeration.
Indeed, to enumerate the maximal bicliques, we need at least to write each of them into the output, which is achieved in $\mathcal{O}\left(q \cdot \B\right)$, our algorithm is therefore a factor $\frac{1}{r} \cdot (\du + \ddu) \cdot d$ away from this value. 
We will evaluate in Section~\ref{sec:experiments} typical values of the factor $\frac{1}{r}$ on real data and see that it is close to 1 in general.

Note that Chen~\etal. also give an expression of the complexity of \oombea as a function of its output~\cite{chen2022efficient}: it runs in~$\mathcal{O}\left(\zu \cdot m \cdot \B\right)$, so comparing these two complexities leads to comparing the factors $\frac{1}{r} \cdot (\du + \ddu) \cdot d \cdot q$ to $\zu \cdot m$, which is not trivial, so we perform an extensive experimental comparison of the running time of both algorithms in the next section.

\section{Experiments}
\label{sec:experiments}

In this section, we perform experiments on the \bbk algorithm to demonstrate its practical efficiency. 
We have implemented this algorithm in \cpp, and the code is available online\,\footnote{\url{https://gitlab.lip6.fr/baudin/bbk}}. 
Throughout this section, the bipartite graphs $G=(U,V,E)$ used are such that the set $U$ is the one containing the \emph{fewest} vertices and the set $V$ is the one containing the \emph{most}.
Unless specified otherwise, we initialize the algorithm on the set $U$.

We first present the bipartite graphs that are used in these experiments, and compare the execution times of \bbk algorithm on these graphs with those of \oombea algorithm~\cite{chen2022efficient}. 
Then, we discuss the influence of the choice of the set on which we initialize the enumeration on the vertices (Line~\ref{line:bbk:forU} of Algorithm~\ref{algo:bbk}).
Finally, we show that although our implementation leads to shorter execution times, the \oombea algorithm is usually more economical in terms of memory consumption.

\subsection{Datasets}

We tested the \bbk algorithm on a set of bipartite graphs retrieved from the KONECT~\cite{data-dnc} database. 
We selected bipartite graphs from different real-life situations, corresponding to various numbers of vertices and edges, in order to test the algorithm in different scenarios.

\paragraph{Type of data.}
The bipartite graphs that we use are presented in Table~\ref{tab:data},
together with  some of their relevant properties.
Some of them concern users actions on online platforms: tags posted in \bibsonomy and \citeulike, books rated on \bookcrossing, movies rated on \movie, \dvdciao, \filmtrust and \wikilens, posts made on forums in \ucforum. 
Other graphs link people to their activity: \actormovie is a graph linking actors to the movies that they have starred in, \citeseer links scientific authors to their publications, \github links users to the projects they are working on. 
Finally, the remaining graphs correspond to various types of information classification: \dailykos, \reuters and \nips connect documents and the words that they contain, \dbpedia associates athletes to their teams, \tvtropes links artistic works to their style, \discogs links musical content to its style, \marvel links Marvel comics characters to the publications in which they appear, \pics connects people to the images on which they are tagged, and \youtube connects users to the groups to which they belong.

Graphs in Table~\ref{tab:data} are sorted by increasing number of maximal bicliques (column $\boldsymbol{\B}$). 
As mentioned above, the two sets of vertices $U$ and $V$ of these bipartite graphs $G = (U,V,E)$ are chosen in such a way that $U$ contains the fewest vertices ($\nu \leq \nv$).
\dvdciao is the graph containing the most bicliques that could be enumerated within a week of computation,
while no algorithm was able to terminate for \nips and \movie
within this computation time limit. 
It is worth noticing that there is no simple relation between the number of maximal bicliques and the number of edges or vertices.
For example, the bipartite graph \filmtrust has numbers of edges and vertices of the same order of magnitude as the graph \wikilens, but it contains more than $200$ times more maximal bicliques.

\begin{table}[!hbtp]
  \centering
  \resizebox{\linewidth}{!}
  {
    \begin{tabular}{|r|rrrr|l|}
      \hline
      \textbf{Graph} & $\boldsymbol{m}$ & $\boldsymbol{\nu}$ & $\boldsymbol{\nv}$ & $\boldsymbol{\B}$ & \textbf{Source: {\footnotesize{\url{http://konect.cc/networks/}}}} \\
      \hline
      \emph{UC-Forum} & \np{7089} & \np{522} & \np{899} & \np{16261} & {\footnotesize \tt opsahl-ucforum} \\
      \emph{Discogs} & \np{481661} & \np{15} & \np{270771} & \np{17650} & {\footnotesize \tt discogs\_lgenre} \\
      \emph{CiteSeer} & \np{512267} & \np{105353} & \np{181395} & \np{171354} & {\footnotesize \tt komarix-citeseer} \\
      \emph{Marvel} & \np{96662} & \np{6486} & \np{12942} & \np{206135} & {\footnotesize \tt marvel} \\
      \emph{DBpedia} & \np{1366466} & \np{34461} & \np{901130} & \np{517943} & {\footnotesize \tt dbpedia-team} \\
      \emph{Actor-Movie} & \np{1470404} & \np{127823} & \np{383640} & \np{1075444} & {\footnotesize \tt actor-movie} \\
      \emph{Pics} & \np{997840} & \np{17122} & \np{495402} & \np{1242718} & {\footnotesize \tt pics\_ui} \\
      \emph{YouTube} & \np{293360} & \np{30087} & \np{94238} & \np{1826587} & {\footnotesize \tt youtube-groupmemberships} \\
      \emph{WikiLens} & \np{26937} & \np{326} & \np{5111} & \np{2769773} & {\footnotesize \tt wikilens-ratings} \\
      \emph{BookCrossing} & \np{1149739} & \np{105278} & \np{340523} & \np{54458953} & {\footnotesize \tt bookcrossing\_full-rating} \\
      \emph{GitHub} & \np{440237} & \np{56519} & \np{120867} & \np{55346398} & {\footnotesize \tt github} \\
      \emph{DailyKos} & \np{353160} & \np{3430} & \np{6906} & \np{242384960} & {\footnotesize \tt bag-kos} \\
      \emph{FilmTrust} & \np{35494} & \np{1508} & \np{2071} & \np{646318495} & {\footnotesize \tt librec-filmtrust-ratings} \\
      \emph{CiteULike} & \np{538761} & \np{22715} & \np{153277} & \np{2333281521} & {\footnotesize \tt citeulike-ut} \\
      \emph{Reuters} & \np{978446} & \np{19757} & \np{38677} & \np{10071287092} & {\footnotesize \tt gottron-reuters} \\
      \emph{BibSonomy} & \np{453987} & \np{5794} & \np{204673} & \np{10526275315} & {\footnotesize \tt bibsonomy-2ut} \\
      \emph{TV-Tropes} & \np{3232134} & \np{64415} & \np{87678} & \np{19636996096} & {\footnotesize \tt dbtropes-feature} \\
      \emph{DVD-Ciao} & \np{1625480} & \np{21019} & \np{71633} & \np{109769732096} & {\footnotesize \tt librec-ciaodvd-review\_ratings} \\
      \nips & \np{746316} & \np{1500} & \np{12375} & - & {\footnotesize \tt bag-nips} \\
      \movie & \np{1000009} & \np{3760} & \np{6040} & - & {\footnotesize \tt movielens-1m} \\
      \hline
    \end{tabular}
  }
  \caption{Datasets used in the experiments, sorted by increasing number of maximal bicliques. 
    $\boldsymbol{m}$ is the number of links in the bipartite graph,
    $\boldsymbol{\nu}$ the number of vertices in set $U$, $\boldsymbol{\nv}$ the number of vertices in its set $V$ (denominated such that $\boldsymbol{\nu} \leq \boldsymbol{\nv}$), and $\boldsymbol{\B}$ the number of maximal bicliques. 
    A ``-'' symbol means that we do not know the number of maximal bicliques, as no algorithm finishes in less than a week for these graphs.}
  \label{tab:data}
\end{table}

\paragraph{Bidegeneracies of graphs and vertices.}

Table~\ref{tab:dmax} presents the maximum degrees $ \boldsymbol{\du} $ and $ \boldsymbol{\dv} $ in the bipartite graphs of Table~\ref{tab:data}, as well as their bidegeneracies $ \boldsymbol{\bu} $ and $ \boldsymbol{\bv} $ defined in Section~\ref{sec:algorithm}. 
We also report the maximum degrees of the graph projected onto $U$ or $V$ ($\boldsymbol{\ddu}$ or $\boldsymbol{\ddv}$), and the average bidegeneracies $\boldsymbol{\overline{\bu}}$ and $\boldsymbol{\overline{\bv}}$ that are respectively the mean of the bidegeneracies of the vertices of $U$ and of $V$.
We remind that the complexities of the \bbk algorithm have been expressed in Theorem~\ref{thm:bbk1} with $\bu$ (or similarly $\bv$), thanks to the nodes processed in a bidegeneracy order.
With a random order, this factor would be bounded by $\ddu + \du$ (or $\ddv + \dv$).
We observe that, while $ \bu $ (and $ \bv$) $\geq max(\du,\dv) $, we almost always have $ \bu = max(\du,\dv) $.
Finally, $ \bu $ (and $ \bv $) are lower than $ \ddu + \du $ (or $ \ddv + \dv $) but typically of the same order of magnitude.
So, the bidegeneracies of the graphs, which appear in the complexity expressions, do not give a clear understanding of the computational gain that the bidegeneracy order brings.
Moreover, when considering the average values of the vertex bidegeneracies ($\boldsymbol{\overline{\bu}}$ and $\boldsymbol{\overline{\bv}}$), we can see that the average vertex  bidegeneracy is often significantly lower than the bidegeneracy of the graph.
As noticed in Section~\ref{subsec:complexity_input}, the complexity as a function of the input can be more precisely expressed as $\mathcal{O}\left(\underset{i=1}{\overset{\nu}{\sum}} \left( (b(u_i) + |X_i|) \cdot 3^{b(u_i)/3} \right)\right)$, where the index $i$ is given by a bidegeneracy order and $X_i$ is defined at Line~\ref{line:bbk:Xui} of Algorithm~\ref{algo:bbk}.
So the fact that vertex bidegeneracies are relatively lower gives a clue as to how well the algorithm works in practice on these graphs: the worst cases of bidegeneracy are not called up on many vertices.

\begin{table}[!hbtp]
  \centering
  \begin{tabular}{|r|rrrrrrrr|}
    \hline
    \textbf{Graph} & $\boldsymbol{\du}$ & $\boldsymbol{\dv}$ & $\boldsymbol{\ddu}$ & $\boldsymbol{\ddv}$ & $\boldsymbol{\bu}$ & $\boldsymbol{\bv}$& $\boldsymbol{\overline{\bu}}$ & $\boldsymbol{\overline{\bv}}$\\
    \hline
    \emph{UC-Forum} & \np{126} & \np{99} & \np{411} & \np{634} & \np{126} & \np{144} & \np{91} & \np{95}\\
    \emph{Discogs} & \np{128070} & \np{15} & \np{15} & \np{270771} & \np{128070} & \np{128070} & \np{32110} & \np{102088}\\
    \emph{CiteSeer} & \np{286} & \np{385} & \np{596} & \np{1653} & \np{385} & \np{385} & 9.3 & \np{43}\\
    \emph{Marvel} & \np{1625} & \np{111} & \np{1934} & \np{9855} & \np{1625} & \np{1625} & \np{40} & \np{769}\\
    \emph{DBpedia} & \np{2671} & \np{17} & \np{2839} & \np{18517} & \np{2671} & \np{2671} & \np{47} & \np{522}\\
    \emph{Actor-Movie} & \np{294} & \np{646} & \np{7799} & \np{3957} & \np{646} & \np{646} & \np{171} & \np{47}\\
    \emph{Pics} & \np{7810} & \np{335} & \np{7079} & \np{113079} & \np{7810} & \np{7810} & \np{173} & \np{1459}\\
    \emph{YouTube} & \np{7591} & \np{1035} & \np{7357} & \np{37514} & \np{7591} & \np{7591} & \np{124} & \np{1049}\\
    \emph{WikiLens} & \np{1721} & \np{80} & \np{285} & \np{4826} & \np{1721} & \np{1721} & \np{138} & \np{1062}\\
    \emph{BookCrossing} & \np{13601} & \np{2502} & \np{53916} & \np{151646} & \np{13601} & \np{13601} & \np{215} & \np{1669}\\
    \emph{GitHub} & \np{884} & \np{3675} & \np{15995} & \np{29650} & \np{3675} & \np{3675} & \np{525} & \np{103}\\
    \emph{DailyKos} & \np{457} & \np{2123} & \np{430} & \np{6895} & \np{2817} & \np{2123} & \np{2808} & \np{1349}\\
    \emph{FilmTrust} & \np{244} & \np{1044} & \np{1459} & \np{1770} & \np{1100} & \np{1044} & \np{1020} & \np{152}\\
    \emph{CiteULike} & \np{4072} & \np{8814} & \np{18190} & \np{80410} & \np{8814} & \np{8814} & \np{3903} & \np{915}\\
    \emph{Reuters} & \np{380} & \np{19044} & \np{19731} & \np{37716} & \np{19044} & \np{19044} & \np{18632} & \np{287}\\
    \emph{BibSonomy} & \np{21463} & \np{1407} & \np{4614} & \np{159465} & \np{21463} & \np{21463} & \np{584} & \np{6919}\\
    \emph{TV-Tropes} & \np{6507} & \np{12400} & \np{47460} & \np{37494} & \np{12400} & \np{12400} & \np{6487} & \np{2331}\\
    \emph{DVD-Ciao} & \np{34884} & \np{422} & \np{13000} & \np{62027} & \np{34884} & \np{34884} & \np{241} & \np{24195}\\
    \nips & \np{914} & \np{1455} & \np{1500} & \np{12363} & \np{1760} & \np{3312} & \np{1755} & \np{2671}\\
    \movie & \np{3428} & \np{2314} & \np{3660} & \np{6040} & \np{3429} & \np{4998} & \np{2492} & \np{4978}\\
    \hline
  \end{tabular}
  \caption{Characteristics of the bipartite graphs of Table~\ref{tab:data}. $\boldsymbol{\du}$ and $\boldsymbol{\dv}$ are the maximum degrees in $U$ and $V$,
    $\boldsymbol{\ddu}$ and $\boldsymbol{\ddv}$ are the maximum projected degrees on $U$ and $V$, $\boldsymbol{\bu}$ and $\boldsymbol{\bv}$ are the maximum bidegeneracies and $\boldsymbol{\overline{\bu}}$ and $\boldsymbol{\overline{\bv}}$ are the mean bidegeneracies of $U$ and $V$.}
  \label{tab:dmax}
\end{table}

\subsection{Results: computation time}

To evaluate the gain in efficiency in the enumeration of maximal bicliques, we measure the computation times of our implementation and compare them to the ones obtained with the Chen~\etal.~\cite{chen2022efficient} implementation\footnote{\url{https://github.com/S1mpleCod/cohesive_subgraph_bipartite}}. 
We carry out these experiments on machines equipped with 2 Intel Xeon E5645 processors with 12 cores each at $\np{2,4}$~GHz and 128~GB of RAM. 
We set a computation time limit of one week, so that computations that exceed this limit are interrupted.

Figure~\ref{fig:time-bkmbe} presents these computation times on a logarithmic scale and the corresponding numerical values are detailed in Table~\ref{tab:time-soa}
(values ``-'' in the table correspond to computations that did not finish within one week).
Graphs are sorted by increasing number of maximal bicliques. 
We can observe that the computation time generally increases with the number of bicliques, whereas it does not seem to be directly related to the number of vertices, edges, degree or bidegeneracy. 
In most cases, the computation time is lower for \bbk than that for \oombea, and the more bicliques the graph contains, the larger the difference, reaching more than a factor 10 for \filmtrust and \citeulike. 
Finally, for the graphs with the largest number of maximal bicliques, \oombea does not obtain all the maximal bicliques within the time limit of one week of computing.

\begin{figure}[!hbtp]
  \centering
  \includegraphics[width=0.9\linewidth]{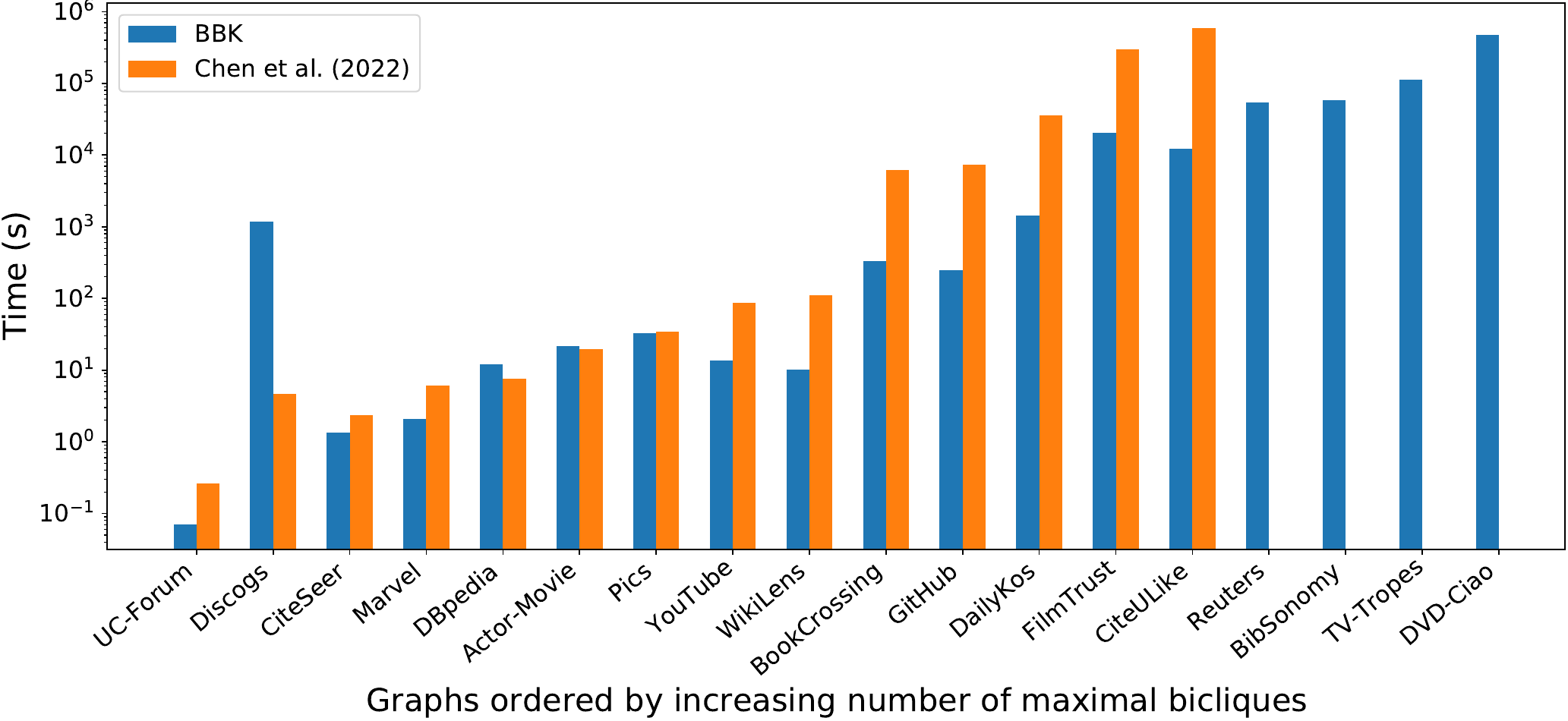}
  \caption{Computation times of the \bbk algorithm on the datasets of Table~\ref{tab:data} compared to those of the \oombea algorithm. 
    On the rightmost graphs, the values for \oombea are not displayed because the computation was not completed within the one week limit of the experiments.}
  \label{fig:time-bkmbe}
\end{figure}

However, there are some graphs (among those with the least maximal bicliques) for which \oombea is faster than \bbk: this is the case for \discogs, \dbpedia and \actormovie. 
\discogs is a unique case in our experiments, for which \oombea is much faster than our algorithm.
It stems from the fact that the structure of this graph is very particular as its set $U$ is small  ($ \nu = 15$),
but it has a high average bidegeneracy ($ \overline{\bu} > 10^4 $), while \bbk tends to be more efficient on graphs with lower bidegeneracy.
By contrast, the \emph{unilateral coreness} value $\zu$ present in Chen~\etal. complexity expression is bounded by the size $\nu$ of $U$, as it is based on the projection onto $U$, which makes \oombea significantly more efficient on this instance.

We also display in Table~\ref{tab:time-soa} the ratio $r$ defined in Section~\ref{subsec:complexity_output}.
We remind that this ratio corresponds to the fraction of leaves in the call tree that return a maximal biclique;
the other leaves are unnecessary for the computation, and would be pruned by an optimal pivot strategy.
We can see that this ratio is relatively close to 1 in our enumerations, which means that a better pivot could not be much more
efficient at pruning useless branches of the call trees. 
Note that three graphs stand out with a lower $r$ values: \dbpedia, \actormovie and \pics. 
Interestingly, these are the cases where \bbk performs relatively poorly by comparison with \oombea in terms of computation time (excluding \discogs discussed above).
This suggests that the reason for this poorer performance
is a lower efficiency of the pivot pruning on these graphs.

\begin{table}[!hbtp]
  \centering
  \setlength{\tabcolsep}{12pt}
  \begin{tabular}{|r|rr||r|}
    \hline
    \textbf{Graph} & $\boldsymbol{t_{\bbk}}$ & $\boldsymbol{t_{\oombea}}$ & $\boldsymbol{r}$\\
    \hline
    \emph{UC-Forum} & 0.07 & 0.26 & 0.60\\
    \emph{Discogs} & 1,185 & 4.6 & 1.00\\
    \emph{CiteSeer} & 1.4 & 2.3 & 0.71\\
    \emph{Marvel} & 2.1 & 6.1 & 0.66\\
    \emph{DBpedia} & 12 & 7.7 & 0.40\\
    \emph{Actor-Movie} & 21 & 19 & 0.21\\
    \emph{Pics} & 32 & 34 & 0.39\\
    \emph{YouTube} & 13 & 87 & 0.76\\
    \emph{WikiLens} & 10 & 111 & 0.91\\
    \emph{BookCrossing} & 335 & 6,169 & 0.77\\
    \emph{GitHub} & 248 & 7,283 & 0.84\\
    \emph{DailyKos} & 1,419 & 35,837 & 0.77\\
    \emph{FilmTrust} & 20,255 & 300,307 & 0.98\\
    \emph{CiteULike} & 12,338 & 594,549 & 0.92\\
    \emph{Reuters} & 54,045 & - & 0.82\\
    \emph{BibSonomy} & 58,719 & - & 0.96\\
    \emph{TV-Tropes} & 113,659 & - & 0.73\\
    \emph{DVD-Ciao} & 476,686 & - & 0.93\\
    \nips & - & - & -\\
    \movie & - & - & -\\
    \hline
  \end{tabular}
  \setlength{\tabcolsep}{6pt} 
  \caption{Computation times (in seconds) obtained by \bbk and \oombea algorithms. A ``-'' symbol means that the computation has not been completed within one week. The last column represents the ratio $r$ for \bbk defined in Section~\ref{subsec:complexity_output} and which appears in the expression of the complexity (Theorem~\ref{thm:bbk2}).}
  \label{tab:time-soa}
\end{table}

\subsection{Starting the enumeration from $U$ or $V$}
\label{subsec:UV}

In Algorithm~\ref{algo:bbk}, we can choose to run the loop on Line~\ref{line:bbk:forU} on the vertices of $U$ or on the vertices of $V$.
While the output of the algorithm is identical, we show here that this choice may have a significant impact on the computation time.

By default, the iterations are carried out on the set $U$ containing the lowest number of vertices, and we use this case as a reference.
Figure~\ref{fig:VU} shows the impact of iterating on the set containing the most vertices, which is the set $V$ in Table~\ref{tab:data}.
The bars represented in this figure correspond to the time  of a run on $V$ divided by the time by a run on $U$, so a bar above $y=1$ means that the run is slower when iterating on the larger set of vertices, and vice versa.

\begin{figure}[!hbtp]
  \centering

  \includegraphics[width=0.9\linewidth]{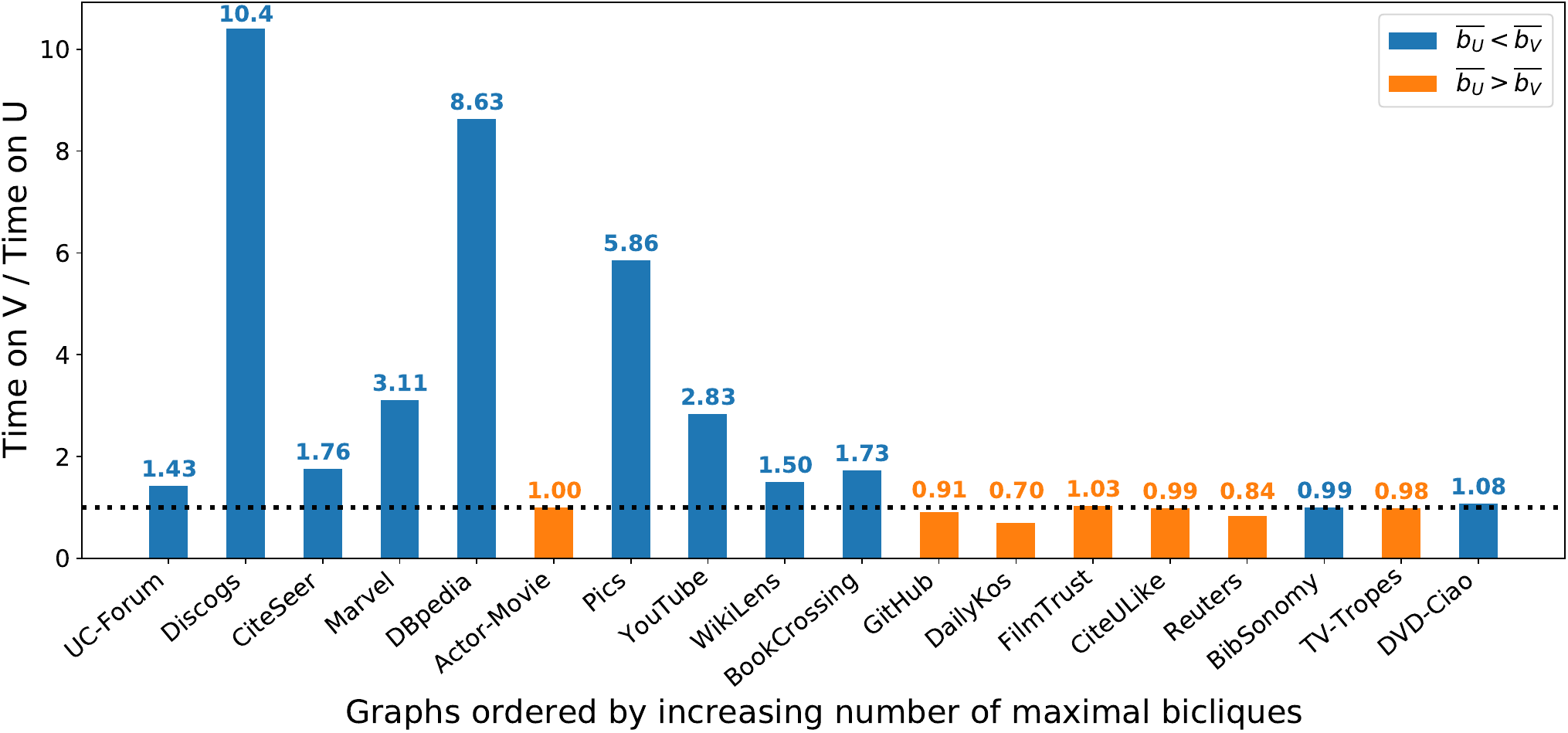}
  
  \caption{Ratio between the execution time of Algorithm~\ref{algo:bbk} \bbk when the run is performed on the larger set $V$ to the execution time when 
    the run is performed on the smaller set $U$. The results in blue correspond to graphs where $\overline{\bu} < \overline{\bv}$, and in orange to graphs where $\overline{\bu} > \overline{\bv}$.}
  \label{fig:VU}
\end{figure}

First, we notice that this choice can have a strong impact on the computation time, since for some graphs such as \discogs, \dbpedia or \pics, the computation time varies by a factor larger than $5$,
while it seems to be less significant for the graphs containing the most maximal bicliques, that have a factor closer to 1.
Then, we observe that choosing the smallest set $U$ is appropriate for two-thirds of the datasets, but there are several graphs where it is more efficient to perform the run on $V$. 
We have seen that if a vertex $u$ has a bidegeneracy $b(u)$, then the time spent in the call of function \BKMBE made at the iteration of the loop of Line~\ref{line:bbk:forU} corresponding to this vertex is exponential in $b(u)$ (see the proof of Theorem~\ref{thm:bbk1}).
Consequently, we investigate if there is a relation between the computation times in regard to $\overline{\bu}$ and $\overline{\bv}$ to understand the origin of these observations.

Thus, we distinguish between two cases: we color a bar in Figure~\ref{fig:VU} in orange when $\overline{\bu} > \overline{\bv}$ and in blue when $ \overline{\bv} > \overline{\bu} $. 
It appears that, almost in all cases, initializing on the set with the lowest mean bidegeneracy is the most efficient choice; 
the exceptions are \filmtrust and \bibsonomy, for which the choice between $U$ or $V$ has little impact on the computation time.

\subsection{About the memory usage}

While \bbk algorithm is in general more efficient than \oombea in terms of computation time, the \oombea implementation proposed by Chen~\etal.~\cite{chen2022efficient} is more economical when considering memory usage. 
To illustrate this, Figure~\ref{fig:mem} shows the memory used by \bbk and \oombea in logarithmic scale on the datasets of Table~\ref{tab:data}.

\begin{figure}[!hbtp]
  \centering
  \includegraphics[width=0.9\linewidth]{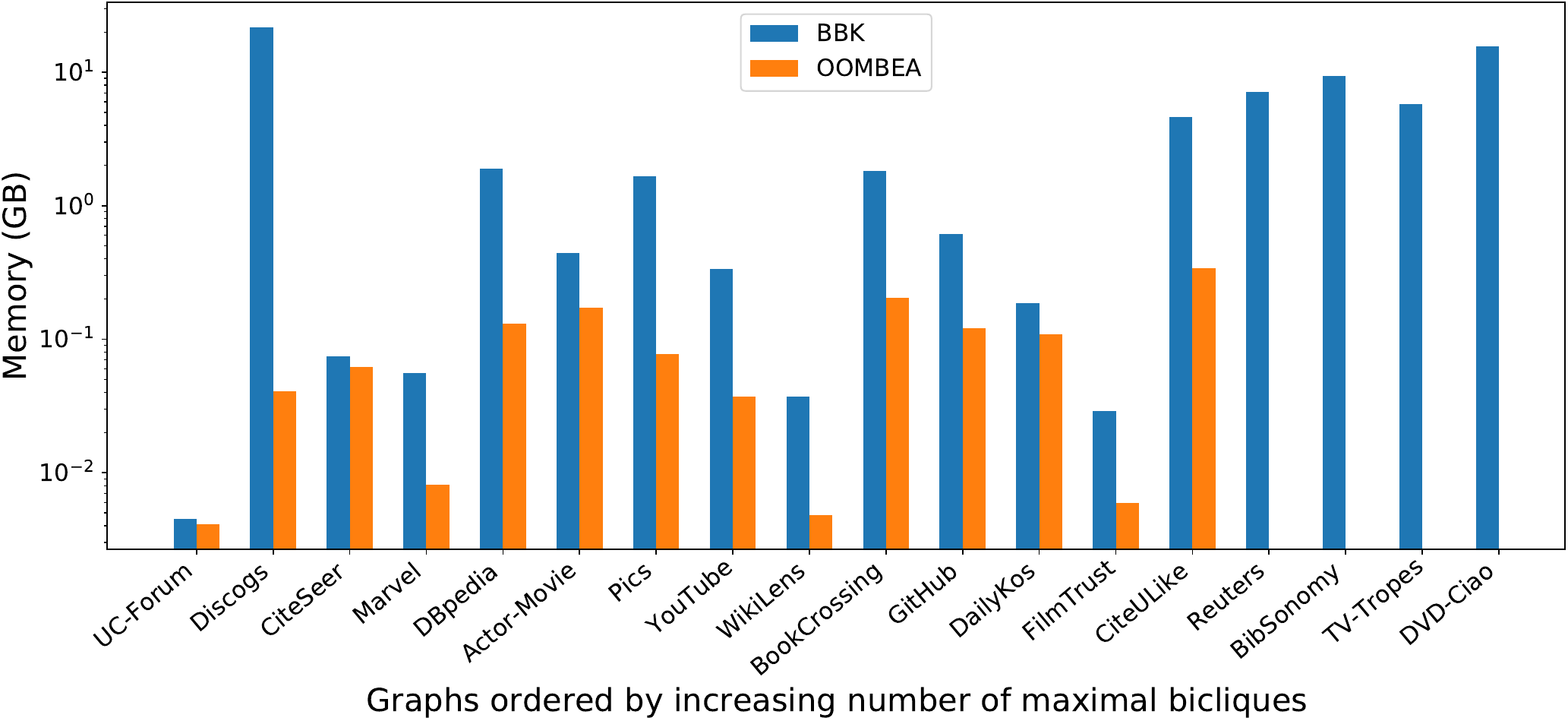}
  \caption{Memory used by the two algorithms \bbk and \oombea on the datasets of Table~\ref{tab:data}. For the four rightmost graphs, \oombea cannot complete the enumeration in less than one week, so its result is not displayed.}
  \label{fig:mem}
\end{figure}

Figure~\ref{fig:mem} shows that \oombea is able to enumerate the maximal bicliques using typically ten times less memory than \bbk (even more for the \discogs special case).
One explanation for this is that to gain efficiency within a recursive call, we use the method described by Eppstein~\etal.~\cite{eppstein2010listing}:
it consists in pre-allocating for each vertex a list of the size of its bidegeneracy.
As a vertex cannot belong to a biclique containing more vertices than its bidegeneracy, this size is a bound on the depth of the tree of calls to \BKMBE related to this vertex.
These lists are then used to record the end index of each vertex adjacency list used on a given recursive call:
it allows reducing the size of the used neighborhoods while avoiding wasting time copying them,
although it adds a factor $n_U \cdot \overline{b_U}$ to the memory complexity of the algorithm.
In bipartite graphs, the average bidegeneracy of vertices can be relatively large (see Table~\ref{tab:dmax}), which implies that storing this list may result in high memory requirements.

Note however that except for five graphs, the memory used by \bbk does not exceed 2~GB, and never exceeds about 20~GB, which is largely manageable on most modern computers.
So, in practice, memory is not the limiting factor for the enumeration of maximal bicliques, which is why we favor an algorithm aiming at time-efficiency.

\section{Conclusion}
\label{sec:conclusion}

In this article, we have introduced the \bbk algorithm for enumerating maximal bicliques in bipartite graphs, which aims at being time-efficient.
To do this, we have adapted the recursive Bron-Kerbosch algorithm to the context of bipartite graphs by using an extended graph on which finding cliques is equivalent to finding bicliques in the original instance. 
Moreover, we take advantage of the sparsity of real-world graphs by formulating the algorithm so as to use the neighborhood in the original bipartite graph only.
To improve its efficiency, the algorithm processes the vertices of the graph in an order that we call bidegeneracy order, which aims at reducing the set of candidate vertices at each recursive call.
We also add to the process a classic pivot-based pruning strategy, adapted to the context of bipartite graphs.
We have carried out a theoretical analysis of the \bbk algorithm to establish two complexity expressions: one as a function of its input and one as a function of its output characteristics.
Finally, we provide an open-source C++ implementation of \bbk, which we have used to illustrate the good performances experimentally on massive real-world datasets. These experiments shown that \bbk can enumerate maximal bicliques typically 10 times faster than the state of the art does on larger instances, 
and produces results in cases where the state of the art is unable to provide a solution within one week of computation.

We identify several directions which can be developed from this work.
One of them is the search for bicliques in non-bipartite graphs.
This question has been explored extensively probably because bicliques play an important role in the structure of real-world graphs such as protein interaction networks~\cite{liu2008searching}.
Indeed, as our approach essentially uses the neighbors and second neighbors of a node, it should be translatable to this context, yet concepts such as the vertex bidegeneracy order would have to be adapted accordingly.
Another interesting lead comes from the fact that finding bicliques in large bipartite graphs is similar to detecting closed itemsets in transaction databases, as mentioned earlier~\cite{makino2004new,zaki1998theoretical}.
Precisely, if we map items to set $U$ and transactions to $V$, an itemset with support larger than $s$ would be a subset of $U$ that forms a biclique with at least $s$ vertices of $V$.
Thus, adapting \bbk to this specific issue may bring new, efficient solutions to this problem.

\section*{Acknowledgments}

This work is funded in part by the ANR (French National Agency of Research) through the FiT-ANR-19-LCV1-0005 grant.

\bibliographystyle{abbrv}
\bibliography{biblio}

\end{document}